\documentclass{amsart}

\usepackage{enumitem}

\usepackage{hyperref}

\usepackage{graphicx}
\usepackage{accents}

\usepackage{xcolor}

\usepackage{amsmath, amssymb, amsfonts, amsthm, fullpage}
\usepackage{graphics,graphicx}
\usepackage{accents}
\usepackage{color}

\usepackage{algorithm}
\usepackage{algorithmicx}
\usepackage{algpseudocode}

\usepackage[margin=0.9in]{geometry}

\newcommand{\norma}[1]{{\left\vert\kern-0.25ex\left\vert\kern-0.25ex\left\vert #1
    \right\vert\kern-0.25ex\right\vert\kern-0.25ex\right\vert}}

\newcommand{\teta}{\tilde{\eta}}

\newcommand{\bx}{\mathbf{x}}
\newcommand{\bn}{\mathbf{n}}
\newcommand{\bu}{\mathbf{u}}

\newcommand{\tDiv}{\tilde{\nabla}\!\cdot\!}
\newcommand{\Div}{\nabla\!\cdot\!}
\newcommand{\Curl}{\nabla\!\times\!}

\newcommand{\tbn}[1]{{\left\vert\kern-0.25ex\left\vert\kern-0.25ex\left\vert #1 \right\vert\kern-0.25ex\right\vert\kern-0.25ex\right\vert}}

\newtheorem{remark}{Remark}[section]
\newtheorem{lemma}{Lemma}[section]
\newtheorem{proposition}{Proposition}[section]
\newtheorem{theorem}{Theorem}[section]

\newtheorem{definition}{Definition}[section]

\AtBeginDocument{%
  \setlength{\belowdisplayskip}{5pt} \setlength{\belowdisplayshortskip}{5pt}
  \setlength{\abovedisplayskip}{5pt} \setlength{\abovedisplayshortskip}{5pt}
}

\begin{document}

\title[Equations for small amplitude shallow water waves over small bathymetric variations]{Equations for small amplitude shallow water waves over small bathymetric variations}


\author{Samer Israwi}
\address{\textbf{S.~Israwi:} Department of Mathematics, Faculty of Sciences 1, Lebanese University, Beirut, Lebanon}
\email{s\_israwi83@hotmail.com}

\author{Youssef Khalifeh}
\address{\textbf{Y. Khalifeh} Laboratory of Mathematics, Faculty of Sciences 1, Lebanese University, Beirut, Lebanon}
\email{khalifeyoussef78@gmail.com}

\author{Dimitrios Mitsotakis}
\address{\textbf{D.~Mitsotakis:} Victoria University of Wellington, School of Mathematics and Statistics, PO Box 600, Wellington 6140, New Zealand}
\email{dimitrios.mitsotakis@vuw.ac.nz}




\date{\today}



\begin{abstract}

A generalized version of the $abcd$-Boussinesq class of systems is derived to accommodate variable bottom topography in two-dimensional space. This extension allows for the conservation of suitable energy functionals in some cases and enables the description of water waves in closed basins with well-justified slip-wall boundary conditions. The derived systems possess a form that ensures their solutions adhere to important principles of physics and mathematics. By demonstrating their consistency with the Euler equations and estimating their approximation error, we establish the validity of these new systems. Their derivation is based on the assumption of small bathymetric variations. With practical applications in mind, we assess the effectiveness of some of these new systems through comparisons with standard benchmarks. The results indicate that the assumptions made during the derivation are not overly restrictive. The applications of the new systems encompass a wide range of scenarios, including the study of tsunamis, tidal waves, and waves in ports and lakes.

\end{abstract}

\maketitle

\section{Introduction}

The propagation of water waves is mathematically described by partial differential equations derived by Leonard Euler (1707-1783) in 1757 \cite{Euler1757} for the study of incompressible flow of an ideal fluid. These equations are known as the (incompressible) Euler equations. To present the Euler equations, we first define the horizontal and vertical coordinates of three-dimensional space as $\mathbf{x}=(x,y)\in \mathbb{R}^2$ and $z$, respectively. The undisturbed water level at rest is considered to be at $z=0$. The velocity of the fluid is denoted as $(\mathbf{u}(\mathbf{x},z,t),w(\mathbf{x},z,t))=(u(\mathbf{x},z,t),v(\mathbf{x},z,t),w(\mathbf{x},z,t))$, separating the vertical component from its horizontal components. The free surface elevation above its undisturbed level is represented by $\eta(\mathbf{x},t)$, and the depth of the ocean floor is denoted as $D(\mathbf{x},t)$. (Note that we consider time-dependent bottom topography to allow for the modeling of water waves generated by the motion of the ocean floor.) The horizontal gradient is denoted by $\nabla = (\partial_x,\partial_y)$, while $t$ represents time. The Euler equations consist of equations for the conservation of mass and momentum, which, for $-D(\mathbf{x},t)<z<\eta(\mathbf{x},t)$, can be written as:
\begin{align}
&\Div\bu+w_z=0\ , \label{eq:mass} \\
&\bu_t+(\bu\cdot\nabla)\bu+w\bu_z+\frac{1}{\rho}\nabla p=0, \ , \label{eq:momentu}\\
&w_t+(\bu\cdot\nabla) w+ww_z+\frac{1}{\rho} p_z+g=0\ . \label{eq:momentv}
\end{align}
Since our aim is to model water waves in oceans bounded above by a free surface and below by the ocean floor, Euler's equations require the following boundary conditions at the free surface 
\begin{equation}
\eta_t+\bu\cdot\nabla\eta-w=0,\quad \text{on $z=\eta(\bx,t)$}\ ,
\end{equation}
and at the bottom
\begin{equation}
D_t+\bu\cdot\nabla(z+D(\bx,t))+w=0\ , \quad\text{for $z=-D(\bx,t)$}\ .
\end{equation}
Another important ingredient in the theory of water waves is the irrotationality condition
\begin{equation}\label{eq:irrotational}
\Curl (\bu,w)=\begin{pmatrix} w_y-v_z\\ u_z-w_x\\ v_x-u_y \end{pmatrix}=0\ ,
\end{equation}
which is unavoidable when potential flow is assumed.
While the pressure inside the fluid volume is an unknown function, the pressure at the free surface is a given constant \begin{equation}\label{eq:pressbc}
p=p_{\text{atm}}, \quad \text{for $z=\eta(\bx,t)$}\ .
\end{equation}

The study of these particular equations presents a great difficulty due to the fact that the fluid domain $\Omega_t$, which includes the free surface of the water and the impermeable bottom (and possibly solid wall boundaries), is practically unknown since it is bounded by the unknown free-surface $z=\eta(\mathbf{x},t)$. Furthermore, the domain is not stationary, adding further complexity to the situation. As a result, scientists have derived approximate equations based on simplifying assumptions. The development of such approximations began after John Scott Russell (1808--1882) discovered solitary waves and published his {\em report on Waves} in 1844, \cite{Russell1844}, and this work has continued to the present day. Examples of these approximations are the Boussinesq systems, which describe long waves of small amplitude. The first Boussinesq system was derived by J. Boussinesq himself \cite{Bous1871,Bous1872}, and subsequent improvements, also referred to as Boussinesq systems, have been derived. It is important to note that there are only a few water wave systems that are mathematically and physically justified. For instance, the Euler equations have been proven to be well-posed (in the Hadamard's sense) only in unbounded domains \cite{wu1997,wu1999,L2005}, while among all the Boussinesq systems that have been derived, only the system proposed in \cite{IKKM2021} has been proven to be well-posed in a bounded domain with slip-wall boundary conditions.

Several Boussinesq systems have been extensively used in the literature to describe nonlinear and dispersive waves with variable bottom topography. One of the notable Boussinesq systems is the Peregrine system \cite{Pere1967}:
\begin{equation}
\label{eq:Peregrin}
\begin{aligned}
&\eta_t+\nabla\cdot[(D+\eta)\bu]=0\ , \\
&{\bf u}_t +g\nabla\eta+(\bu\cdot \nabla)\bu-\frac{1}{2}D\nabla(\nabla\cdot(D\bu_t))+\frac{1}{6}D^2\nabla(\nabla\cdot\bu_t)=0\ ,
\end{aligned}
\qquad \text{(Peregrine)}
\end{equation}
and the Nwogu (or extended Boussinesq) system \cite{Nwogu93}
\begin{equation}\label{eq:Nwogu}
\begin{aligned}
& \eta_t+\Div((D+\eta)\bu)+\Div[\bar{a}D^2\nabla(\Div(D\bu))+\bar{b}D^3\nabla(\Div\bu)]=0\ ,\\
& \bu_t+g\nabla\eta+(\bu\cdot\nabla)\bu-[\bar{c}D\nabla(\Div(D\bu_t))+\bar{d}D^2\nabla(\Div\bu_t)]=0\ ,
\end{aligned}\qquad \text{(Nwogu)}
\end{equation}
with $\bar{a}=\theta-1/2$, $\bar{b}=1/2[(\theta-1)^2-1/3]$, $\bar{c}=1-\theta$, $\bar{d}=1/2(1-\theta)^2$, and $\theta=1/5$. (The parameter $\theta$ determines the depth in which the horizontal velocity $\bu$ is computed, cf. e.g. (\ref{eq:thetad})). Although in some cases their Cauchy problems have been proved well-posed, \cite{DuchIsr2018,mits2009,BLIG2022}, they do not preserve any meaningful approximation of the total energy of the Euler equations. While the particular property is not restrictive, it is rather desirable from physical as well as numerical point of view. Specifically, we can ensure physical validity for long periods of time when the solution of a system preserves an energy functional \cite{Feng2010}. Similarly, systems of BBM-BBM-type \cite{mits2009} or the systems with flat bottom topography of \cite{BCL2005} fail at describing water waves in a basin with slip-wall boundary conditions \cite{dms2009,DMS2010,DMS2007}. Some other Boussinesq systems even fail in having traveling waves of finite energy, proving that asymptotic justification does not imply physical relevance, \cite{BDM2007i,BDM2008ii}. For all these reasons in \cite{KMS2020, IKKM2021} we derived and analyzed a new Boussinesq system of BBM-BBM-type which can be used in a straightforward manner in bounded domains with slip-wall boundary conditions. Moreover, its solutions preserve the same energy functional as its non-dispersive counterpart, i.e. the shallow water-waves (Saint-Venant) equations \cite{Whitham2011} making the system attractive from numerical and physical point of view. 

In general, Boussinesq systems are extensions of equations that describe the propagation of water waves in one direction, such as the Korteweg-de Vries (KdV) equation. Various versions of the KdV equation have been derived to account for variable bottom topography \cite{Isrtal13,Israwii10,DurIsrawi12}. On the other hand, such models cannot describe wave reflections on walls or obstacles. Moreover, high-order Boussinesq equations, including the Serre-Green-Naghdi equations \cite{Serre, GN1976, Lannes13}, have only been proven to be well-posed in unbounded domains \cite{Israwi11}. Note that the Serre-Green-Naghdi equations and their optimized counterpart of \cite{CDM2017} preserve an energy functional.

In this article, we start by presenting in Section \ref{sec:derivation} the asymptotic derivation of a new class of Boussinesq systems that generalizes the so-called $abcd$-Boussinesq systems of \cite{BCL2005,BCS2002} in the case of variable bottom topography. We also establish their consistency with the Euler equations along the lines of \cite{BCL2005} in Section \ref{consistwithEulerr}. These systems also generalize the systems of \cite{Chen03} in two-dimensional domains, and their derivation is based on the same assumption of small bottom variations. Although we do not provide any well-posedness results, these new systems appear to have analogous theoretical properties to the system of \cite{BCL2005}, and they can further be used with slip-wall boundary conditions. The well-posedness of the BBM-BBM system was proved in \cite{IKKM2021}, while the well-posedness of other Boussinesq systems of the same family, such as the Bona-Smith systems, is very similar. The new system that corresponds to the Nwogu system can be applied with accurate slip-wall boundary conditions, as demonstrated in Section \ref{sec:bcs}, and is validated against laboratory data in Section \ref{sec:valid}. For the Nwogu system as well as for its regularized counterpart, we have justified this set of boundary conditions in one-dimensional domains in \cite{mm2023}. Taking into account the results of \cite{IKKM2021,KMS2020}, we conclude that the assumption of small bottom variations is not restrictive, and the necessary nonlinear and dispersive behaviors of the Boussinesq systems are retained. This was also demonstrated in \cite{MS1991}. For the sake of completeness, we present a variational derivation of the new energy-conserving $abcd$-Boussinesq systems in Section \ref{sec:energy}. The new systems have similar energy conservation properties to the $abcd$-Boussinesq systems of \cite{BCL2005,BCS2002}, and thus we generalize them in the case of variable bottom topography.

\section{Derivation of the new Boussinesq equations}\label{sec:derivation}

Consider the depth function $D(\bx,t)=D_0+D_b(\bx)+\zeta(\bx,t)$, where $D_0$ is a characteristic (mean) depth and $D_b$ is smooth and varies gently in $\mathbb{R}^2$. The displacement of the bottom's moving part, denoted by $\zeta$, is small compared to the depth and does not exceed the amplitude of the free surface elevation. If $\lambda_0$ is a typical wavelength, $a_0$ a typical wave height, and $d_0$ a typical order of bottom topography variations, then we consider the (scaled) non-dimensional independent variables 
\begin{equation}\label{eq:indnondim}
\tilde{\bx}=\frac{\bx}{\lambda_0},\quad \tilde{z}=\frac{z}{D_0}, \quad \tilde{t}=\frac{c_0}{\lambda_0}t\ ,
\end{equation}
and also the non-dimensional dependent variables
\begin{equation}\label{eq:depnondim}
\tilde{\bu}=\frac{D_0}{a_0c_0}\bu, \quad \tilde{w}=\frac{\lambda_0}{a_0c_0}w,\quad \tilde{\eta}=\frac{\eta}{a_0},\quad \tilde{D}_b=\frac{D_b}{d_0}\ ,\quad \tilde{p}=\frac{p}{\rho g D_0}\ , \quad \tilde{\zeta}=\frac{\zeta}{a_0}\ ,
\end{equation}
where $c_0=\sqrt{gD_0}$ is the linear speed of propagation. Denoting as usual
$$\varepsilon = \frac{a_0}{D_0},\quad \sigma=\frac{D_0}{\lambda_0}, \quad \beta=\frac{d_0}{D_0}\ ,$$
and assuming that $0<\varepsilon\approx \sigma^2\lesssim \beta\ll 1$ all the aforementioned Boussinesq systems are asymptotically equivalent, in the sense that all of them occur from the Euler equations after appropriate scaling and discarding high-order terms with the same order in $\varepsilon$ and $\sigma^2$. In the new variable the scaled depth will be  $\tilde{D}=1+\beta \tilde{D}_b+\varepsilon\tilde{\zeta}$. Note that $\tilde{D}_{\tilde{t}}=O(\varepsilon)$ and thus we assume a slowly varying bottom. The requirement $\beta\ll 1$ is used to describe small bottom variations compared to the wavelength.

Using the new variables, we write Euler's equations (\ref{eq:mass})--(\ref{eq:momentv}) in non-dimensional form:
\begin{align}
& \tilde{\nabla} \cdot \tilde{\bu}+\tilde{w}_{\tilde{z}}=0\ , \label{eq:eul1}\\
&\varepsilon \tilde{\bu}_{\tilde{t}}+\varepsilon^2 [(\tilde{\bu}\cdot\tilde{\nabla})\tilde{\bu}+\tilde{w}\tilde{\bu}_{\tilde{z}}]+\tilde{\nabla}\tilde{p}=0\ ,  \label{eq:eul2}\\
&\varepsilon\sigma^2 \tilde{w}_{\tilde{t}}+\varepsilon^2\sigma^2 [\tilde{\bu}\cdot \tilde{\nabla} \tilde{w}+\tilde{w}\tilde{w}_{\tilde{z}}]+\tilde{p}_{\tilde{z}}=-1\ ,  \label{eq:eul3}
\end{align}
for $-\tilde{D}<\tilde{z}<\varepsilon\tilde{\eta}$. The irrotationality condition is also transformed to
\begin{align}
&\tilde{\nabla}\times \tilde{\bu}=0\ , \label{eq:eul4}\\
&\bu_z-\sigma^2\tilde{\nabla}\tilde{w}=0\ ,   \label{eq:eul5}
\end{align}
for $-\tilde{D}<\tilde{z}<\varepsilon\tilde{\eta}$. The boundary conditions on the free surface and the bottom are written as
\begin{align}
& \teta_t+\varepsilon(\tilde{\bu}\cdot\tilde{\nabla}\teta)-\tilde{w}=0,\quad \tilde{p}=\frac{p_{\text{atm}}}{\rho g D_0}\quad \text{on}\quad \tilde{z}=\varepsilon\teta\ ,  \label{eq:eul6}\\
& \tilde{\zeta}_{\tilde{t}}+ \tilde{\bu}\cdot \tilde{\nabla}\tilde{D}+\tilde{w}=0\quad \text{on}\quad\tilde{z}=-\tilde{D}\ .  \label{eq:eul7}
\end{align}

First, integrate the mass equation (\ref{eq:eul1}) between $-\tilde{D}$ and $\varepsilon\tilde{\eta}$ to obtain
\begin{equation}
\tilde{w}(\varepsilon\teta)-\tilde{w}(-\tilde{D})=-\int_{-\tilde{D}}^{\varepsilon\teta}\tilde{\nabla}\cdot\tilde{\bu}~d\tilde{z}\ .
\end{equation}
Using the boundary conditions (\ref{eq:eul6}) and (\ref{eq:eul7}) we obtain the equation
\begin{equation}\label{eq:massper}
\teta_{\tilde{t}}+\tilde{\nabla}\cdot[(\tilde{D}+\varepsilon\tilde{\eta})\tilde{\bu}_a]+\tilde{\zeta}_t=0\ ,
\end{equation}
where 
\begin{equation}\label{eq:dpav}
\tilde{\bu}_a(\tilde{\bx},\tilde{t})=\frac{1}{\tilde{D}+\varepsilon\teta}\int_{-\tilde{D}}^{\varepsilon\teta}\tilde{\bu}~d\tilde{z}\ ,
\end{equation}
denotes the depth-averaged horizontal velocity of the fluid. Equation (\ref{eq:massper}) is exact and it is the mass equation of the hydrostatic shallow water wave equations and of Peregrine's original system.

Integrating (\ref{eq:eul1}) from $-\tilde{D}$ to $\tilde{z}$, and using (\ref{eq:eul7}) we have
\begin{equation}\label{eq:eul8}
\tilde{w}=-\tilde{\bu}\cdot \tilde{\nabla}\tilde{D}-\int_{-\tilde{D}}^{\tilde{z}}\tilde{\nabla}\cdot\tilde{\bu}-\tilde{\zeta}_{\tilde{t}}\ .
\end{equation}
After integration of (\ref{eq:eul5}) over $(0,\tilde{z})$, we have
$
\tilde{\bu}(\tilde{\bx},\tilde{z},\tilde{t})=\tilde{\bu}_0(\tilde{\bx},\tilde{t})+O(\sigma^2)\ ,
$
where $\tilde{\bu}_0(\tilde{\bx},\tilde{t})\doteq \tilde{\bu}(\tilde{\bx},0,\tilde{t})$ denotes the horizontal velocity at the bottom. Similarly, integration of (\ref{eq:eul5}) over $(-\tilde{D},\tilde{z})$ yields
$ \tilde{\bu}(\tilde{\bx},\tilde{z},\tilde{t})=\tilde{\bu}_b(\tilde{\bx},\tilde{t})+O(\sigma^2)$, where $\tilde{\bu}_b(\tilde{\bx},\tilde{t})\doteq \tilde{\bu}(\tilde{\bx},-\tilde{D},\tilde{t})$. Thus, we have 
\begin{equation}\label{eq:eul9}
\tilde{\bu}(\tilde{\bx},\tilde{z},\tilde{t})=\tilde{\bu}_0(\tilde{\bx},\tilde{t})+O(\sigma^2)=\tilde{\bu}_b(\tilde{\bx},\tilde{t})+O(\sigma^2)\ .
\end{equation}
Substitution of (\ref{eq:eul8}) into the irrotationality condition (\ref{eq:eul5}) and using (\ref{eq:eul9}) yields
\begin{equation}\label{eq:eul10}
\tilde{\bu}_{\tilde{z}}=-\sigma^2(\tilde{z}+1) \tilde{\nabla}(\tilde{\nabla}\cdot \tilde{\bu}_0)-\sigma^2\tilde{\nabla}\tilde{\zeta}_{\tilde{t}}+O(\sigma^4,\varepsilon\sigma^2,\beta\sigma^2)\ .
\end{equation}
Integration of (\ref{eq:eul10}) from $0$ to $\tilde{z}$ implies
\begin{equation}\label{eq:eul11}
\tilde{\bu}=\tilde{\bu}_0-\sigma^2(\tilde{z}+\frac{\tilde{z}^2}{2})\tilde{\nabla}(\tilde{\nabla}\cdot \tilde{\bu}_0)-\sigma^2\tilde{z}\nabla\tilde{\zeta}_{\tilde{t}}+O(\sigma^4,\varepsilon\sigma^2,\beta\sigma^2)\ .
\end{equation}
Equation (\ref{eq:eul8}) using (\ref{eq:eul9}) becomes
\begin{equation}\label{eq:eul12}
\tilde{w}=-\tilde{\nabla}\cdot(\tilde{D}\tilde{\bu}_0)-\tilde{z}\tilde{\nabla}\cdot \tilde{\bu}_0-\tilde{\zeta}_{\tilde{t}}+O(\sigma^2)\ ,
\end{equation}
and differentiation with respect to $t$ yields
\begin{equation}\label{eq:eul13}
\tilde{w}_{\tilde{t}}=-\tilde{\nabla}\cdot(\tilde{D}\tilde{\bu}_0)_{\tilde{t}}-\tilde{z}\tilde{\nabla}\cdot {\tilde{\bu}_0}_{\tilde{t}}-\tilde{\zeta}_{\tilde{t}\tilde{t}}+O(\sigma^2)\ .
\end{equation}
Setting $\tilde{P}=\tilde{p}-p_{\text{atm}}/\rho g D_0$ (so as $\tilde{\nabla} \tilde{P}=\tilde{\nabla} \tilde{p}$ and $\tilde{P}(\varepsilon\teta)=0$) and integrating (\ref{eq:eul3}) from $\tilde{z}$ to $\varepsilon\tilde{\eta}$ we obtain
\begin{equation}\label{eq:eul14}
\tilde{P}=\varepsilon\sigma^2(\tilde{z}+\frac{\tilde{z}^2}{2})\tilde{\nabla}\cdot{\tilde{\bu}_0}_{\tilde{t}}+\varepsilon\sigma^2\tilde{z}\tilde{\zeta}_{\tilde{t}\tilde{t}}+\varepsilon\tilde{\eta}-\tilde{z}+O(\varepsilon\sigma^4,\varepsilon^2\sigma^2,\varepsilon\beta\sigma^2)\ .
\end{equation}
Substitution of (\ref{eq:eul11}), (\ref{eq:eul12}) and (\ref{eq:eul14}) into (\ref{eq:eul2}) leads to the approximation of momentum conservation
\begin{equation}\label{eq:eul15}
{\tilde{\bu}_0}_{\tilde{t}}+\tilde{\nabla}\tilde{\eta}+\varepsilon(\tilde{\bu}_0\cdot\tilde{\nabla})\tilde{\bu}_0=O(\sigma^4,\varepsilon\sigma^2,\beta\sigma^2)\ .
\end{equation}
Equation (\ref{eq:eul11}) using (\ref{eq:dpav}) becomes
\begin{equation}\label{eq:eul17}
\tilde{\bu}_0=\tilde{\bu}_a-\frac{\sigma^2}{3}\tilde{\nabla}(\tilde{\nabla}\cdot {\tilde{\bu}_a})-\frac{\sigma^2}{2}\tilde{\nabla}\tilde{\zeta}_{\tilde{t}}+O(\sigma^4,\varepsilon\sigma^2,\beta\sigma^2)\ .
\end{equation}
Subsequently, equation (\ref{eq:eul15}) yields the momentum equation
\begin{equation}\label{eq:eul18}
{\tilde{\bu}_a}_{\tilde{t}}+\tilde{\nabla}\tilde{\eta}+\varepsilon({\tilde{\bu}_a}\cdot\tilde{\nabla}){\tilde{\bu}_a}-\frac{\sigma^2}{3}\tilde{\nabla}(\tilde{\nabla}\cdot{\tilde{\bu}_a}_{\tilde{t}})-\frac{\sigma^2}{2}\tilde{\nabla}\tilde{\zeta}_{\tilde{t}\tilde{t}}=O(\sigma^4,\varepsilon\sigma^2,\beta\sigma^2)\ .
\end{equation}
Since $\tilde{\bu}=\tilde{\bu}_a+O(\sigma^2)$, it is implied that $\tilde{\nabla}\times \tilde{\bu}_a=O(\sigma^2)$, which yields that $({\tilde{\bu}_a}\cdot\tilde{\nabla}){\tilde{\bu}_a}=\frac{1}{2}\tilde{\nabla}|\tilde{\bu}_a|^2+O(\sigma^2)$. Therefore, we can further simplify (\ref{eq:eul18}) into 
\begin{equation}\label{eq:eul19}
{\tilde{\bu}_a}_{\tilde{t}}+\tilde{\nabla}\tilde{\eta}+\frac{\varepsilon}{2}\tilde{\nabla}|\tilde{\bu}_a|^2-\frac{\sigma^2}{3}\tilde{\nabla}(\tilde{\nabla}\cdot{\tilde{\bu}_a}_{\tilde{t}})-\frac{\sigma^2}{2}\tilde{\nabla}\tilde{\zeta}_{\tilde{t}\tilde{t}}=O(\sigma^4,\varepsilon\sigma^2,\beta\sigma^2)\ .
\end{equation}
Note that equations (\ref{eq:massper})-(\ref{eq:eul19}) have been also studied in \cite{KMS2020}. Evaluating the horizontal velocity (\ref{eq:eul11}) at depth \begin{equation}\label{eq:thetad}
\tilde{z}_\theta=-\tilde{D}+\theta(\varepsilon\tilde{\eta}+\tilde{D}) \quad\text{for} \quad 0\leq \theta\leq 1\ ,
\end{equation}
where $\sigma^2\tilde{z}_\theta=(\theta-1)\sigma^2+O(\varepsilon\sigma^2,\beta\sigma^2)$, and using (\ref{eq:eul17}) we obtain
\begin{equation}\label{eq:eul20}
\tilde{\bu}_\theta=\tilde{\bu}_a-\frac{\sigma^2}{2}\left(\theta^2-\frac{1}{3}\right)\tilde{\nabla}(\tilde{\nabla}\cdot\tilde{\bu}_a)-\sigma^2\left(\theta-\frac{1}{2}\right)\tilde{\nabla}\tilde{\zeta}_{\tilde{t}}+O(\sigma^4,\varepsilon\sigma^2,\beta\sigma^2)\ ,
\end{equation}
and equivalently
\begin{equation}\label{eq:eul21}
\tilde{\bu}_a=\tilde{\bu}_\theta+\frac{\sigma^2}{2}\left(\theta^2-\frac{1}{3}\right)\tilde{\nabla}(\tilde{\nabla}\cdot\tilde{\bu}_\theta)+\sigma^2\left(\theta-\frac{1}{2}\right)\tilde{\nabla}\tilde{\zeta}_{\tilde{t}}+O(\sigma^4,\varepsilon\sigma^2,\beta\sigma^2)\ .
\end{equation}
Substituting (\ref{eq:eul21}) into (\ref{eq:massper}) and (\ref{eq:eul19}) we obtain the system
\begin{align}
&\teta_{\tilde{t}}+\tilde{\nabla}\cdot[(\tilde{D}+\varepsilon\tilde{\eta})\tilde{\bu}_\theta]+\frac{\sigma^2}{2}\left(\theta^2-\frac{1}{3}\right)\tilde{\nabla}\cdot\tilde{\nabla}(\tilde{\nabla}\cdot\tilde{\bu}_\theta)+\sigma^2\left(\theta-\frac{1}{2}\right)\tilde{\nabla}\cdot\tilde{\nabla}\tilde{\zeta}_{\tilde{t}}+\tilde{\zeta}_t=O(\sigma^4,\varepsilon\sigma^2,\beta\sigma^2)\ , \label{eq:eul22}\\
&{\tilde{\bu}_\theta}_{\tilde{t}}+\tilde{\nabla}\tilde{\eta}+\frac{\varepsilon}{2}\tilde{\nabla}|\tilde{\bu}_\theta|^2+\frac{\sigma^2}{2}(\theta^2-1)\tilde{\nabla}(\tilde{\nabla}\cdot{\tilde{\bu}_\theta}_{\tilde{t}})+\sigma^2(\theta-1)\tilde{\nabla}\tilde{\zeta}_{\tilde{t}\tilde{t}}=O(\sigma^4,\varepsilon\sigma^2,\beta\sigma^2)\ . \label{eq:eul23}
\end{align}
Observe that from (\ref{eq:eul22}) and (\ref{eq:eul23}) we obtain
\begin{equation}\label{eq:eul24}
\tilde{\nabla}\cdot \tilde{\bu}_\theta=-\teta_{\tilde{t}}-\tilde{\zeta}_{\tilde{t}}+O(\varepsilon,\beta,\sigma^2)\ ,
\end{equation}
and
\begin{equation}\label{eq:eul25}
{\tilde{\bu}_{\theta}}_{\tilde{t}}=-\tilde{\nabla}\tilde{\eta}+O(\varepsilon,\sigma^2)\ . 
\end{equation}
Using the classical BBM-trick \cite{BBM1972,Pere1966} with (\ref{eq:eul24}), (\ref{eq:eul25}) and taking arbitrary $\nu,\mu\in \mathbb{R}$ we write 
\begin{equation}\label{eq:rel1}
\tilde{\nabla}\cdot\tilde{\nabla}(\tilde{\nabla}\cdot\tilde{\bu}_\theta)=-\tilde{\nabla}\cdot\tilde{\nabla}~\teta_{\tilde{t}}-\tilde{\nabla}\cdot\tilde{\nabla}~{\tilde{\zeta}}_{\tilde{t}}+O(\varepsilon,\beta,\sigma^2)\ ,
\end{equation} 
and 
\begin{equation}\label{eq:rel2}
\tilde{\nabla}(\tilde{\nabla}\cdot{\tilde{\bu}_\theta}_{\tilde{t}})=-\tilde{\nabla}(\tilde{\nabla}\cdot \tilde{\nabla}\tilde{\eta})+O(\varepsilon,\sigma^2)\ .
\end{equation} 
Substituting relations (\ref{eq:rel1})--(\ref{eq:rel2}) into the system (\ref{eq:eul22})--(\ref{eq:eul23}) we obtain the general $abcd$-Boussinesq system
\begin{align}
& \teta_{\tilde{t}}+\tilde{\nabla}\cdot[(\tilde{D}+\varepsilon\tilde{\eta})\tilde{\bu}_\theta]+\sigma^2 \tilde{\nabla}\cdot[a\tilde{\nabla}(\tilde{\nabla}\cdot\tilde{\bu}_\theta)-b\tilde{\nabla}\teta_{\tilde{t}}]+\tilde{a}\sigma^2 \tilde{\nabla}\cdot\tilde{\nabla}\tilde{\zeta}_{\tilde{t}}+\tilde{\zeta}_{\tilde{t}}=O(\sigma^4,\varepsilon\sigma^2,\beta\sigma^2)\ , \label{eq:eul26}\\
& {\tilde{\bu}_\theta}_{\tilde{t}}+\tilde{\nabla}\tilde{\eta}+\frac{\varepsilon}{2}\tilde{\nabla}|\tilde{\bu}_\theta|^2+\sigma^2 \tilde{\nabla}[ c \tilde{\nabla}\cdot \tilde{\nabla}\teta-d\tilde{\nabla}\cdot{\tilde{\bu}_\theta}_{\tilde{t}}]+\sigma^2\tilde{c}\tilde{\nabla}\tilde{\zeta}_{\tilde{t}\tilde{t}}=O(\sigma^4,\varepsilon\sigma^2,\beta\sigma^2)\ .\label{eq:eul27}
\end{align}
where $\tilde{a}=\mu(\theta-1/2)-(1-\mu)(1/2[(\theta-1)^2-1/3])$, $\tilde{c}=(\theta-1)$ and $a,b,c,d$ as in (\ref{eq:abcdcoef2}).

Note that 
\begin{align}
& \tilde{\nabla}(\tilde{\nabla}\cdot\tilde{\bu}_\theta)=\tilde{\nabla}(\tilde{D}^3\tilde{\nabla}\cdot\tilde{\bu}_\theta)+O(\varepsilon,\beta)\ ,\\
&\tilde{\nabla}\teta_{\tilde{t}}=\tilde{D}^2\tilde{\nabla}\teta_{\tilde{t}}+O(\varepsilon,\beta)\ , \\
&\tilde{\nabla}\cdot \tilde{\nabla}\teta=\tilde{\nabla}\cdot (\tilde{D}^2 \tilde{\nabla}\teta)+O(\varepsilon,\beta)\ ,\\
&\tilde{\nabla}\cdot{\tilde{\bu}_\theta}_{\tilde{t}}=\tilde{\nabla}\cdot (\tilde{D}^2{\tilde{\bu}_\theta}_{\tilde{t}})+O(\varepsilon,\beta)\ ,\\
&\tilde{\nabla}\cdot\tilde{\nabla}\tilde{\zeta}_{\tilde{t}}=\tilde{\nabla}\cdot(\tilde{D}^2\tilde{\nabla}\tilde{\zeta}_{\tilde{t}})+O(\varepsilon,\beta)\ ,\\
&\tilde{\nabla}\tilde{\zeta}_{\tilde{t}\tilde{t}}=\tilde{D}\tilde{\nabla}\tilde{\zeta}_{\tilde{t}\tilde{t}}+O(\varepsilon,\beta)\ .
\end{align}
Using the previous approximations we write the general $abcd$-Boussinesq system as
\begin{align}
& \teta_{\tilde{t}}+\tilde{\nabla}\cdot[(\tilde{D}+\varepsilon\tilde{\eta})\tilde{\bu}_\theta]+\sigma^2 \tilde{\nabla}\cdot[a\tilde{\nabla}(\tilde{D}^3\tilde{\nabla}\cdot\tilde{\bu}_\theta)-b\tilde{D}^2\tilde{\nabla}\teta_{\tilde{t}}]+\tilde{a}\sigma^2 \tilde{\nabla}\cdot(\tilde{D}^2\tilde{\nabla}\tilde{\zeta}_{\tilde{t}})+\tilde{\zeta}_{\tilde{t}}=O(\sigma^4,\varepsilon\sigma^2,\beta\sigma^2)\ , \label{eq:eul28}\\
& {\tilde{\bu}_\theta}_{\tilde{t}}+\tilde{\nabla}\tilde{\eta}+\frac{\varepsilon}{2}\tilde{\nabla}|\tilde{\bu}_\theta|^2+\sigma^2 \tilde{\nabla} [ c \tilde{\nabla}\cdot (\tilde{D}^2 \tilde{\nabla}\teta)-d\tilde{\nabla}\cdot (\tilde{D}^2{\tilde{\bu}_\theta}_{\tilde{t}})]+\sigma^2\tilde{c}\tilde{D}\tilde{\nabla}\tilde{\zeta}_{\tilde{t}\tilde{t}}=O(\sigma^4,\varepsilon\sigma^2,\beta\sigma^2)\ .\label{eq:eul29}
\end{align}
The reason for considering such approximations is to ensure that  certain combinations of the parameters $a,b,c,d$ lead to systems of significant interest of \cite{BCS2002} that  preserve meaningful approximations of the total energy even with variable bottom topography. 

Discarding the high-order terms in (\ref{eq:eul28})--(\ref{eq:eul29}) we write the equations (\ref{eq:eul28})--(\ref{eq:eul29}) in dimensional variables as 
\begin{equation}\label{eq:Nwogunabcdd2}
\begin{aligned}
& \eta_t+\Div[(D+\eta)\bu]+\Div\left\{a\nabla(D^3\Div\bu)-bD^2\nabla\eta_t\right\}=-\tilde{a}\Div(D^2\nabla\zeta_t)-\zeta_t\ ,\\
& \bu_t+g\nabla\eta+\tfrac{1}{2}\nabla|\bu|^2 +\nabla\left\{cg\Div(D^2\nabla\eta)-d\Div(D^2\bu_t)\right\}=-\tilde{c}D\nabla\zeta_{tt}\ ,
\end{aligned}
\end{equation}
where
\begin{equation}\label{eq:abcdcoef2}
\begin{aligned}
&a=\frac{1}{2}\left(\theta^2-\frac{1}{3}\right)\mu,\quad b=\frac{1}{2}\left(\theta^2-\frac{1}{3}\right)(1-\mu), \quad c=\frac{1}{2}(1-\theta^2)\nu, \quad d= \frac{1}{2}(1-\theta^2)(1-\nu)\ , \\
&\tilde{a}=\mu\left(\theta-\frac{1}{2}\right)-(1-\mu)\left(\frac{1}{2}\left[(\theta-1)^2-\frac{1}{3}\right]\right), \quad \tilde{c}=(\theta-1)\ ,
\end{aligned}
\end{equation}
for fixed $\mu,\nu\in \mathbb{R}$ and $0\leq \theta\leq 1$. 
The first equation corresponds to an approximation of the mass conservation law, while the second equation corresponds to an approximation of the momentum conservation law.
Note that the system (\ref{eq:Nwogunabcdd2}) has been derived in such a way that the $\Curl\bu$ is preserved for all times $t\geq 0$ for stationary bottom topography, and thus the irrotationality of the water waves is respected exactly. Moreover, we observe that in most cases the presence of the regularization operators $I-b\Div(D^2\nabla\bullet)$ and $I-d\nabla\Div(D^2\bullet)$ requires for well-posedness that $b,d\geq 0$ \cite{DM2008}. As we shall see later, conservation of energy requires additionally $b=d\geq 0$.

Water waves are dispersive by their nature. For modeling purposes we employ the dispersion relation of the linearized equations and with flat bottom topography as a measure of accuracy. Specifically, to study the dispersion characteristics of the system  (\ref{eq:Nwogunabcdd2}) we first consider the expansion of the linear dispersion relation for general periodic solutions of the Euler equations and of the form $e^{i(kx-\omega t)}$, which is
\begin{equation}\label{eq:eulerdisprel}
\frac{c^2_{\text{Euler}}}{gD}=\frac{\tanh(Dk)}{Dk}=1-\frac{1}{3}(Dk)^2+\frac{2}{15}(Dk)^4+O\left((Dk)\right)^6\ ,
\end{equation}
where $k$ is the wavenumber and $\omega$ the frequency of the wave. Recall that the period of a linear wave is defined as $T=2\pi/\omega$, the phase speed is $\omega/k$, the linear  speed of propagation is $\sqrt{gD}$ and the wavelength is $\lambda=2\pi/k$. The system (\ref{eq:Nwogunabcdd2}) with flat bottom $D$ has linear dispersion relation given by the formula
 \begin{equation}\label{eq:abcdisperel}
\begin{aligned}
\frac{c^2_{abcd}}{gD}&=\frac{(1-a(Dk)^2)(1-c(Dk)^2)}{(1+b(Dk)^2)(1+d(Dk)^2)}\\
&=1-(a+b+c+d)(Dk)^2+\left((a + b) (b + c) + (a + b + c) d + d^2\right)(Dk)^4+O\left((Dk)\right)^6\ .
\end{aligned}
\end{equation}
If we choose the coefficients $a,b,c,d$ of (\ref{eq:abcdcoef2}) such that the formula (\ref{eq:abcdisperel}) coincides with (\ref{eq:eulerdisprel}) up to the term of order $(Dk)^4$, then we obtain the optimized extended Boussinesq system for $\nu=0$, $\mu=1$ and $\theta^2=1/5$. This is the system which we call Nwogu system and is practically the only system of the $abcd$ class of systems with such an optimal dispersive behavior. Notable systems of the $abcd$ class that have appeared in various works \cite{BC1998,BCS2002,BS1976,DM2008} are the following:
\begin{itemize}[leftmargin=*]
\item The Peregrine system ($a=b=c=0$, $d=1/3$, i.e, $\nu=\mu=0$, $\theta^2=1/3$)
\item The BBM-BBM system ($a=c=0$, $b=d=1/6$, i.e. $\nu=\mu=0$, $\theta^2=2/3$)
\item The Bona-Smith systems ($a=0$, $b=d=(3\theta^2-1)/6$, $c=(2-3\theta^2)/3$, $2/3<\theta^2\leq1$, i.e. $\nu=0$, $\mu=(4-6\theta^2)/3(1-\theta^2)$)
\item The KdV-KdV system ($a=c=1/6$, $b=d=0$, i.e. $\nu=\mu=1$, $\theta^2=2/3$)
\item The Nwogu system ($a=-1/15$, $b=c=0$, $d=2/5$, i.e. $\nu=0$, $\mu=1$, $\theta^2=1/5$)
\end{itemize}

Next, we will prove that all these systems are asymptotically equivalent. On the other hand, it worth mentioning that some of these systems cannot be used for the physical description of surface water waves. For example, the KdV-KdV system does not admit classical solitary wave solution (except perhaps embedded ones) but instead has generalized solitary waves that decay exponentially to small periodic orbits. Such solutions have infinite energy and thus they are not physically meaningful \cite{DM2008,BDM2007i,BDM2008ii}. 

\section{Consistency with the Euler equations}\label{consistwithEulerr}

In this section we consider the Boussinesq system (\ref{eq:eul28})--(\ref{eq:eul29}) with stationary bottom topography and parameters $a,b,c,d$ as in (\ref{eq:abcdcoef2})
\begin{equation}\label{eq:Nwogunabcd}
\begin{aligned}
& \eta_{t}+\nabla\cdot[(D+\varepsilon\eta)\bu]+\sigma^2 \Div\left\{a\nabla(D^3\Div\bu)-bD^2\nabla\eta_t\right\}=0\ , \\
& {\bu}_{t}+\nabla\eta+\frac{\varepsilon}{2}\nabla|\bu|^2+\sigma^2 \nabla\left\{c\Div(D^2\nabla\eta)-d\Div(D^2\bu_t)\right\}=0\ .
\end{aligned}
\end{equation}
Note that in the notation of this section we have dropped the tilde in nondimensional variables.

Consider also the Euler equations in the nondimensional and scaled Zakharov formulation:
 \begin{equation}\label{eq:S0}
\begin{aligned}
&\eta_t = \cfrac{1}{\varepsilon}\,Z_{\varepsilon}[\varepsilon\eta,\beta b]\psi\ , \\
& \psi_t + \eta + \cfrac{\varepsilon}{2} |\nabla\psi|^{2} - \cfrac{\varepsilon}{\sigma^2}\cfrac{(Z_{\varepsilon}[\varepsilon\eta,\beta b]\psi + \varepsilon\sigma^2\nabla\eta\cdot\nabla\psi)^2}{2(1 + \varepsilon^{2}\sigma^2|\nabla\eta|^2)} = 0 \ ,
\end{aligned}
\end{equation}
where $\psi$ is the velocity potential and $Z_{\varepsilon}[\varepsilon\eta,\beta b]\psi = \sqrt{1+\varepsilon^2|\sigma^2\nabla\eta|^2}\partial_n\Phi_{|z=\varepsilon\eta}$ is the Dirichlet-Neumann operator. Following  \cite{Lannes13} (Section 5.3) one can derive the following four-parameter family of Boussinesq systems of order $O(\varepsilon\sigma^2,\beta\sigma^2)$
\begin{equation}\label{eq:S1}
\begin{aligned}
&\eta_t + \nabla\cdot[(D+\varepsilon\eta)\bu] + \sigma^2\Div [\mathbf{a} D^{2}\nabla(\Div D\bu) + \mathbf{b} D^{3}\nabla(\Div \bu)] = 0\ , \\
& \bu_t + \nabla\eta + \varepsilon(\bu\cdot\nabla)\bu +\sigma^2\nabla[\mathbf{c} D\Div(D\bu_t) + \mathbf{d} D^{2}\Div\bu_t] = 0 \ ,
\end{aligned}
\end{equation}
where $\bu$ is the horizontal velocity at the level $(D + \varepsilon\eta)z+\varepsilon\eta$ and $z=-1+\frac{1}{\sqrt{3}}\sqrt{1+2(\delta -\theta)}$,\\
    with $$ \bold{a} = \frac{\delta + \lambda}{3},\;\;  \bold{b} = -\frac{\theta + \lambda}{3},\;\; \bold{c} = -\frac{\alpha + \delta - 1}{3},\;\; \bold{d} = \frac{\alpha + \theta}{3}\ .$$
    
Here, we will prove consistency results between the asymptotic model (\ref{eq:Nwogunabcd}) and the Euler equations in case of small amplitude topographic variations $\beta = O(\sigma^2)$. However, note that from \cite{Lannes13} the Zakharov system (\ref{eq:S0}) and the Boussinesq system (\ref{eq:S1}) are consistent (see Remark 5.35 in \cite{Lannes13}). In particular, for $\delta = \frac{3}{2} \theta^{2} + \theta -\frac{1}{2} $, $\lambda = -\frac{3}{2} \theta^{2} +2\theta -1$ and $\alpha = -\frac{3}{2}\theta^{2}+2\theta -\frac{3}{2}$, we get Nwogu’s system 
 \begin{equation}\label{eq:Nwogu1}
\begin{aligned}
& \eta_t+\Div((D+\varepsilon\eta)\bu)+\sigma^2\Div[\bar{a}D^2\nabla(\Div(D\bu))+\bar{b}D^3\nabla(\Div\bu)]=0\ ,\\
& \bu_t+\nabla\eta+\varepsilon(\bu\cdot\nabla)\bu-\sigma^2[\bar{c}D\nabla(\Div(D\bu_t))+\bar{d}D^2\nabla(\Div\bu_t)]=0\ ,
\end{aligned}
\end{equation}
where $\bar{a}=\theta-1/2$, $\bar{b}=1/2[(\theta-1)^2-1/3]$, $\bar{c}=1-\theta$ and $\bar{d}=1/2(1-\theta)^2$.  Next we will show in detail the consistency of the new system (\ref{eq:Nwogunabcd}) with Nwogu system (\ref{eq:Nwogu1}). But we shall first introduce the definition of consistency along the lines of \cite{BCL2005}. 
\begin{definition}\label{defcons}
Let $p, s \in \mathbb{R}$, $\varepsilon_0 > 0$, $T>0$ and let ${(\bu^{\varepsilon},
\eta^{\varepsilon})_{0<\varepsilon<\varepsilon_0}}$ be a family of solutions of a Boussinesq system $(S_1)$ bounded in
$W^{1,\infty}([0,\frac{T}{\varepsilon}];\,H^{p}(\mathbb{R}^2)^{3})$ independently of $\varepsilon$.
This family is called consistent (with regularity $p$ and $s$) with a system $(S_2)$ if
it  satisfies the system $(S_2)$ with a residual of order $\varepsilon^2$ in $L^{\infty}([0,\frac{T}{\varepsilon}];\,H^{s}(\mathbb{R}^2)^{3})$.
\end{definition}

First we show consistency between the classical Nwogu system (\ref{eq:Nwogu}) and the new Nwogu-type system
\begin{equation}\label{eq:Nwogu2}
\begin{aligned}
& \eta_t+\Div((D+\varepsilon\eta)\bu)+\sigma^2a\Div[D^3\nabla(\Div\bu)]=0\ ,\\
& \bu_t+\nabla\eta+\tfrac{\varepsilon}{2}\nabla|\bu|^2 -\sigma^2b\nabla[\Div(D^2\bu_t)]=0\ ,
\end{aligned}
\end{equation}
where $b=\frac{1}{3}-a=\frac{1}{2}(1-\theta^2)>0$, $a<0$. This system is the special case of system (\ref{eq:Nwogunabcd}) for $\nu=0$, $\mu=1$, $0<\theta^2<1/3$.

\begin{lemma}\label{thrm:curl}
 Let $t_0>1$, $s\geq t_0+3$,
$U_0=(\eta_0,\bu_0)\in H^{s}(\mathbb{R}^2)\times H^{s+2}(\mathbb{R}^2)^2$ with $\vert \Curl \bu_0\vert_{H^{s+1}}\leq \sigma^2 C(\vert U_0\vert_{H^s\times H^{s+2}})$. Then,
 the solution $U=(\eta, \bu)\in C([0,\frac{T}{\varepsilon}];\,H^{s}(\mathbb{R}^2)\times H^{s+2}(\mathbb{R}^2)^2)$
of the new Nwogu system of Boussinesq type (\ref{eq:Nwogu2}) with the initial condition $(\eta_0,\bu_0)$ satisfies
$$
\vert \Curl\bu\vert_{L^{\infty}([0,\frac{T}{\varepsilon}],H^{s+1})}\leq\sigma^2\, C(T,\vert U_0\vert_{H^s\times H^{s+2}})\quad \forall\,0<T<T_{max}\ .
$$
\end{lemma}
\begin{proof}
The proof of existence is demonstrated in (\cite[Theorem 1.1]{saut2012cauchy}). Applying the operator $ \hbox{curl}\,( \cdot)$ to the second equation of the system (\ref{eq:Nwogu2}) we see that
$$
\Curl\bu_t= 0\ ,
$$
which implies $\Curl\bu=\Curl \bu_0$, and the result follows.
\end{proof}

\begin{lemma}\label{newformulate}
Let  $\bu$, $\eta$ be solutions either of (\ref{eq:Nwogu1}) or (\ref{eq:Nwogu2}) and $D=1+\beta D_b$ be a bottom parametrization, then
\begin{equation}\label{definewv}
\begin{aligned}
&\Div[D^3\nabla(\Div\bu)]=\Div[D^2\nabla(\Div(D\bu))]+O(\beta)\ ,\\
&\Div(D\bu)=-\eta_t+O(\varepsilon,\sigma^2)\ ,\\
&\bu_t=-\nabla\eta+O(\varepsilon,\sigma^2)\ .
\end{aligned}
\end{equation}
\end{lemma}
\begin{proof}
The first identity follows from the fact that $D(\bx)=1+\beta D_b(\bx)$. The other two follow from the system (\ref{eq:Nwogu2}).
\end{proof}

Here, we present the main result of this section. For simplicity, we take $\beta = \varepsilon = O(\sigma^2)$ for $\sigma^2\ll 1$.
\begin{proposition}\label{propcons}
Let $s$ large enough, $T>0$,  and $(\eta^\varepsilon,\bu^\varepsilon)_{0<\varepsilon<\varepsilon_0}$ be a family of solutions of Nwogu system (\ref{eq:Nwogu1}) bounded in $W^{1,\infty}\big([0,\frac{T}{\varepsilon}]; H^s(\mathbb{R}^2)^3\big)$. If $D(\bx)=1+\beta D_b(\bx)$, then according to the Definition \ref{defcons} this family is consistent with (\ref{eq:Nwogunabcd}).
\end{proposition}
\begin{proof} 
For the sake of simplicity we drop the index $\varepsilon$ from the notation of  $(\eta^\varepsilon,\bu^\varepsilon)$ but we will still assume solutions of (\ref{eq:Nwogu1}). From the first equation of (\ref{eq:Nwogu1}), it suffices to prove that  $$\Div[\bar{a}D^2\nabla(\Div(D\bu)) + \bar{b}D^3\nabla(\Div\bu)] = \Div[a\nabla(D^3\Div\bu) - bD^2\nabla\eta_t]+O(\beta)  $$ 
where $a = A\mu$, $b = A(1-\mu)$, $\mu\in\mathbb{R}$ and $A=\bar{a}+\bar{b}$.

Using Lemma \ref{newformulate} we have 
\begin{align*}
  \Div[\bar{a}D^2\nabla(\Div(D\bu))+\bar{b}D^3\nabla(\Div\bu)] &= \bar{a}\Div[D^3\nabla(\Div\bu)]+\bar{b}\Div[D^3\nabla(\Div\bu)] +O(\beta) \\ 
  &= A(\mu + 1 - \mu)\Div[D^3\nabla(\Div\bu)] +O(\beta) \\
  &= a\Div[D^3\nabla(\Div\bu)] + b\Div[D^3\nabla(\Div\bu)] +O(\beta) \\
   &= a\Div[\nabla(D^3\Div\bu)] + b\Div[D^2\nabla(\Div(D\bu))] +O(\beta) \\
    &= a\Div[\nabla(D^3\Div\bu)] - b\Div[D^2\nabla\eta_t] +O(\beta)\ ,
    \end{align*}
    where $O(\beta) = \beta f$ and $f \in L^{\infty}([0,\frac{T}{\varepsilon}];\,H^{s-5}(\mathbb{R}^2))$.
    Therefore, $$\Div[\bar{a}D^2\nabla(\Div(D\bu))+\bar{b}D^3\nabla(\Div\bu)] = \Div[a\nabla(D^3\Div\bu) - bD^2\nabla\eta_t] \ , $$
  in the sense of consistency.
 Thus, $(\eta^\varepsilon,\bu^\varepsilon)$ is consistent with the first equation in (\ref{eq:Nwogunabcd}) with a residual of order $\beta^2$ in $L^{\infty} ([0,\frac{T}{\varepsilon}];\, H^{s-5}(\mathbb{R}^2))$.
        
  The same situation holds for the second equation of (\ref{eq:Nwogu1}). Here we denote $c = B\nu$, $d = B(1-\nu)$, $\nu\in\mathbb{R}$ and $B=\bar{c}+\bar{d}$. From Lemma \ref{thrm:curl}, and again using Lemma \ref{newformulate} we see that 
  \begin{align*}
      \bar{c}D\nabla(\Div(D\bu_t))+\bar{d}D^2\nabla(\Div\bu_t) &=  \bar{c}\nabla(D\Div(D\bu_t))+\bar{d}\nabla(D^2\Div \bu_t) + O(\beta) \\
      &= \bar{c}\nabla(\Div(D^2\bu_t))+\bar{d}\nabla(\Div ( D^2\bu_t)) + O(\beta) \\
      &= B\nabla(\Div(D^2\bu_t)) + O(\beta) \\
      &= c\nabla(\Div(D^2\bu_t)) + d\nabla(\Div(D^2\bu_t)) + O(\beta) \\
      &= -c\nabla(\Div(D^2\nabla\eta)) + d\nabla(\Div(D^2\bu_t)) + O(\beta) \ ,
  \end{align*} 
where $O(\beta) = \beta g$ and $g \in L^\infty([0,\frac{T}{\varepsilon}];\;H^{s-4}(\mathbb{R}^2)^{2})$. So, 
$$\bar{c}D\nabla(\Div(D\bu_t))+\bar{d}D^2\nabla(\Div\bu_t) = -c\nabla(\Div(D^2\nabla\eta)) + d\nabla(\Div(D^2\bu_t)) \ ,$$
in the sense of consistency. Therefore, $(\eta^\varepsilon,\bu^\varepsilon)$ is consistent with the second equation of (\ref{eq:Nwogunabcd}) with a residual of order $\beta^2$ in $L^{\infty}([0,\frac{T}{\varepsilon}];\,H^{s-4}(\mathbb{R}^2)^{2})$. We conclude that $(\eta^\varepsilon,\bu^\varepsilon)$ is consistent with (\ref{eq:Nwogunabcd}) in the sense of definition \ref{defcons} and the proof is complete.
  \end{proof}
  
The Euler equations are consistent with the System (\ref{eq:Nwogu1}) at order $\varepsilon^2$ in $C([0,\frac{T}{\varepsilon}];\,\dot{H} ^{s+10} \times H^{s+5})$ as it was proved in \cite[Corollary 5.31]{Lannes13}. A direct consequence of this is that the Euler equations are consistent with the new class of systems (\ref{eq:Nwogunabcd}).

\begin{theorem}[Error estimates]
Let $s\geq 0$, $\varepsilon_0>0$, $b\in H^{\infty}(\mathbb{R}^2)$ and let $U_0=(\eta_0^\varepsilon,\psi_0^\varepsilon)_{0<\varepsilon<\varepsilon_0}$ be a bounded family in $H^p(\mathbb{R}^2)\times \dot{H}^{p+1}(\mathbb{R}^2)$($p\geq s$ large enough). Assume moreover that there is $h_{\min}>0$ such that the non-cavitation condition $h=1+\varepsilon\eta_0^\varepsilon+\beta D_b\geq h_{\min}$ is satisfied, and that there exists $T>0$ such that
\begin{itemize}[leftmargin=*]
\item[$-$] There is a unique family of solutions $U=(\eta^\varepsilon_E, \psi^\varepsilon_E)\in C\big([0,\frac{T}{\varepsilon}];\,H^p(\mathbb{R}^2)\times \dot{H}^{p+1}(\mathbb{R}^2)\big)$ of the Euler equations (\ref{eq:S0}) with initial condition $U_0$.
\item[$-$] There is a unique family of solutions $(\eta_B^\varepsilon,\bu_B^\varepsilon)\in C\big([0,\frac{T}{\varepsilon}];\,H^p(\mathbb{R}^2)^3\big)$ of (\ref{eq:Nwogunabcd}) such that $(\eta_B^\varepsilon,\bu_B^\varepsilon)\vert_{(t=0)}=(\eta_0^\varepsilon,\nabla\psi_0^\varepsilon)$.
\end{itemize}
Then, if $(\eta^\varepsilon_E, \bu^\varepsilon_E)_{0<\varepsilon<\varepsilon_0}$ and $(\eta_B^\varepsilon, \bu_B^\varepsilon)_{0<\varepsilon<\varepsilon_0}$ are bounded in $W^{1,\infty}\big([0,\frac{T}{\varepsilon}];\,H^p(\mathbb{R}^2)^3\big)$ w.r.t $\varepsilon$, then for all $0<\varepsilon<\varepsilon_0$ and $t\in \left[0,\frac{T}{\varepsilon}\right]$
\begin{equation}\label{eq:errestim}
|\eta_B^\varepsilon - \eta^\varepsilon_E|_{L^{\infty}([0,t]\times H^s)} + |\bu_B^\varepsilon- \bu^\varepsilon_E|_{L^{\infty}([0,t]\times H^s)} \leq C\;\varepsilon^2t
\end{equation}
where $\bu^\varepsilon_E=\nabla\psi^\varepsilon_E$.
\end{theorem}
\begin{proof}
 The error estimate (\ref{eq:errestim}) follows using the triangle inequality and the results of \cite[Corollary 2]{BCL2005} for both $(\eta^\varepsilon_E, \bu^\varepsilon_E)$ and $(\eta_B^\varepsilon, \bu_B^\varepsilon)$. 

Indeed, performing symmetrization like in \cite{BCL2005}  (using appropriate nonlinear change of variables) to the new Boussinesq systems we obtain a class of systems with symmetric nonlinearity. This is the class of the so-called symmetric systems (see also \cite{BCL2005}), which is denoted by $S_{\theta,\mu,\nu}^{'}$. Let $(\bu,\eta)$ be a solution consistent with $S_{\theta,\mu,\nu}^{'}$ and $(\underline{\theta},\underline{\mu},\underline{\nu})$ fixed satisfying $a=c$. By applying the nonlinear change of variables, we have that $(\bu,\eta)$ is consistent with the class $S_{\underline{\theta},\underline{\mu},\underline{\nu}}^{'}$. Denote by $\Sigma=S_{\underline{\theta},\underline{\mu},\underline{\nu}}^{'}$ the corresponding subclass of symmetric non-linear and dispersive systems.
Moreover, for all $\varepsilon>0$, let $(\bu_{\Sigma}^\varepsilon,\eta_{\Sigma}^\varepsilon)$ be the exact solution of a system of the class $\Sigma$, then
\begin{align*}
|\bu_B^\varepsilon &- \bu^\varepsilon_E|_{L^{\infty}([0,t]\times H^s)} + |\eta_B^\varepsilon - \eta^\varepsilon_E|_{L^{\infty}([0,t]\times H^s)} = \\
&|\bu_B^\varepsilon - (1-\frac{\varepsilon}{2}(1-\theta^{2})\Delta)^{-1}(1-\frac{\varepsilon}{2}(1-\underline{\theta}^{2})\Delta)(\bu_{\Sigma}^\varepsilon(1-\frac{\varepsilon}{2}\eta_{\Sigma}^\varepsilon)) \\
&+ (1-\frac{\varepsilon}{2}(1-\theta^{2})\Delta)^{-1}(1-\frac{\varepsilon}{2}(1-\underline{\theta}^{2})\Delta)(\bu_{\Sigma}^\varepsilon(1-\frac{\varepsilon}{2}\eta_{\Sigma}^\varepsilon)) - \bu^\varepsilon_E|_{L^{\infty}([0,t]\times H^s)} \\
&+ |\eta_B^\varepsilon - \eta^\varepsilon_{\Sigma} + \eta^\varepsilon_{\Sigma} - \eta^\varepsilon_E|_{L^{\infty}([0,t]\times H^s)} \\
&\leq |\bu_B^\varepsilon - (1-\frac{\varepsilon}{2}(1-\theta^{2})\Delta)^{-1}(1-\frac{\varepsilon}{2}(1-\underline{\theta}^{2})\Delta)(\bu_{\Sigma}^\varepsilon(1-\frac{\varepsilon}{2}\eta_{\Sigma}^\varepsilon))|_{L^{\infty}([0,t]\times H^s)} \\
&+ |\eta_B^\varepsilon - \eta^\varepsilon_{\Sigma}|_{L^{\infty}([0,t]\times H^s)} + |\eta^\varepsilon_{\Sigma} - \eta^\varepsilon_E|_{L^{\infty}([0,t]\times H^s)} \\
&+ |(1-\frac{\varepsilon}{2}(1-\theta^{2})\Delta)^{-1}(1-\frac{\varepsilon}{2}(1-\underline{\theta}^{2})\Delta)(\bu_{\Sigma}^\varepsilon(1-\frac{\varepsilon}{2}\eta_{\Sigma}^\varepsilon)) - \bu^\varepsilon_E|_{L^{\infty}([0,t]\times H^s)} \\
& = I + II\ ,
\end{align*}
where 
$$\begin{aligned}
I=&|\bu_B^\varepsilon - (1-\frac{\varepsilon}{2}(1-\theta^{2})\Delta)^{-1}(1-\frac{\varepsilon}{2}(1-\underline{\theta}^{2})\Delta)(\bu_{\Sigma}^\varepsilon(1-\frac{\varepsilon}{2}\eta_{\Sigma}^\varepsilon))|_{L^{\infty}([0,t]\times H^s)} \\ 
&+|\eta_B^\varepsilon - \eta^\varepsilon_{\Sigma}|_{L^{\infty}([0,t]\times H^s)} \ ,
\end{aligned}$$
and
$$\begin{aligned}
II=&|(1-\frac{\varepsilon}{2}(1-\theta^{2})\Delta)^{-1}(1-\frac{\varepsilon}{2}(1-\underline{\theta}^{2})\Delta)(\bu_{\Sigma}^\varepsilon(1-\frac{\varepsilon}{2}\eta_{\Sigma}^\varepsilon)) - \bu^\varepsilon_E|_{L^{\infty}([0,t]\times H^s)} \\
& + |\eta^\varepsilon_{\Sigma} - \eta^\varepsilon_E|_{L^{\infty}([0,t]\times H^s)}\ .
\end{aligned}$$
Using \cite[Corollary 2]{BCL2005} for the solution of the Boussinesq system $(\bu_B^\varepsilon, \eta_B^\varepsilon)$ for all $\varepsilon>0$, we obtain $I \leq C\;\varepsilon^2t$. Now, 
thanks to the Proposition \ref{propcons} and \cite[Corollary 2]{BCL2005} for the solution of the Euler equations $(\bu_E^\varepsilon, \eta_E^\varepsilon)$, for all $\varepsilon>0$, yields $II \leq C\;\varepsilon^2t $, which implies (\ref{eq:errestim}) and completes the proof.
\end{proof}

\begin{remark} We have established that the new Boussinesq systems are asymptotically equivalent. However, Boussinesq systems have different well-posedness requirements for different sets of the parameters $a,b,c,d$. In the previous theorem we consider $p$ large enough so as to embrace all valid Boussinesq systems. The maximal time of existence also varies depending on the particular system but all systems follow the general results of \cite{DMS2007,saut2012cauchy}. The situation is much different when bounded domains are considered. The results then depend on the choice of the system as well as on the boundary conditions. A discussion on justified slip-wall boundary conditions for Boussinesq systems is presented in Section \ref{sec:bcs}.  Furthermore, the shape of the solitary waves and their dynamics such as interactions and collisions depend also on the choice of the parameters $a,b,c,d$.
\end{remark}

\section{Energy conservation and variational derivation}\label{sec:energy}

In this section we consider $\Omega=\mathbb{R}^2$ while the derivations hold true in an bounded domain as well, under the hypothesis that the solutions along with their derivatives vanish appropriately on the boundary. For example, the boundary conditions
$$\tilde{\bu}\cdot\bn=0,\quad a \tilde{\nabla}(\tilde{D}^3\tDiv\tilde{\bu})\cdot\bn=0,\quad b \tilde{\nabla}~\tilde{\eta}\cdot\bn=0\quad \text{ on }\ \partial\Omega\ ,$$
are adequate to carry the following analysis. Once again we consider stationary bottom topography as there is no meaning to talk about energy conservation with a moving bottom.

Consider the system (\ref{eq:eul28})--(\ref{eq:eul29}) with stationary bottom $\tilde{D}$. Setting 
$$\tilde{R}=(\tilde{D}+\varepsilon\tilde{\eta})\tilde{\bu}+\sigma^2\left\{a\nabla(\tilde{D}^3\tDiv\tilde{\bu})-b\tilde{D}^2\tilde{\nabla}{\tilde\eta}_{\tilde{t}}\right\}\ ,$$
and 
$$\tilde{Q}=\tilde{\eta}+\varepsilon\tfrac{1}{2}|\tilde{\bu}|^2 +\sigma^2\left\{c\tDiv(\tilde{D}^2\tilde{\nabla}\tilde{\eta})-d\tDiv(\tilde{D}^2{\tilde{\bu}}_{\tilde{t}})\right\}\ ,$$
and using the boundary conditions we have
$$
0 =\int_\Omega {\tilde{\eta}}_{\tilde{t}}\tilde{Q}+{\tilde{\bu}}_{\tilde{t}}\cdot \tilde{R}+\tDiv \tilde{R}\tilde{Q}+\tilde{R}\cdot \tilde{\nabla} \tilde{Q} = \int_{\Omega} {\tilde{\eta}}_{\tilde{t}}\tilde{Q}+{\tilde{\bu}}_{\tilde{t}}\cdot \tilde{R}\ .
$$
Moreover, if $b=d$ we have
$$
\begin{aligned}
0 &= \int_{\Omega} {\tilde{\eta}}_{\tilde{t}}\tilde{Q}+{\tilde{\bu}}_{\tilde{t}}\cdot \tilde{R}\\
&=\int_{\Omega} {\tilde{\eta}}_{\tilde{t}}\tilde{\eta}+\varepsilon\tfrac{1}{2}{\tilde{\eta}}_{\tilde{t}}|\bu|^2-\sigma^2c\tilde{D}^2\tilde{\nabla}{\tilde{\eta}}_{\tilde{t}}\cdot \tilde{\nabla}\tilde{\eta}+\tilde{D}\tilde{\bu}\cdot{\tilde{\bu}}_{\tilde{t}}+\varepsilon\tilde{\eta} \tilde{\bu}\cdot{\tilde{\bu}}_{\tilde{t}}-\sigma^2a\tilde{D}^3 \tDiv\tilde{\bu}\tDiv {\tilde{\bu}}_{\tilde{t}}\\
&=\frac{1}{2}\frac{d}{d\tilde{t}}\int_\Omega \tilde{\eta}^2+(\tilde{D}+\varepsilon\tilde{\eta})|\tilde{\bu}|^2-\sigma^2\left[c\tilde{D}^2|\tilde{\nabla}\tilde{\eta}|^2+a\tilde{D}^3(\tDiv\tilde{\bu})^2 \right]\ .
\end{aligned}
$$
Therefore, for $a\leq 0$, $c\leq 0$ and $b=d\geq 0$ we define the energy functional
\begin{equation}\label{eq:energy}
E(t)=\frac{1}{2}\int_\Omega \tilde{\eta}^2+(\tilde{D}+\varepsilon\tilde{\eta})|\tilde{\bu}|^2-\sigma^2\left[c\tilde{D}^2|\tilde{\nabla}\tilde{\eta}|^2+a\tilde{D}^3(\tDiv\tilde{\bu})^2 \right]\ ,
\end{equation}
which in dimensional variables can be written as
\begin{equation}\label{eq:energyd}
E(t)=\frac{1}{2}\int_\Omega g\eta^2+(D+\eta)|\bu|^2-\left[cgD^2|\nabla\eta|^2+aD^3(\Div\bu)^2 \right]\ .
\end{equation}
Using the particular energy function, we can derive energy conservative systems of the class (\ref{eq:Nwogunabcdd2}) using the variational method of \cite{CD2012}. We define the kinetic energy as
$$\mathcal{K}=\frac{\rho}{2}\int_{t_1}^{t_2}\int_\Omega (D+\eta)|\bu|^2-aD^3(\Div\bu)^2~d\bx~dt\ ,$$
and the potential energy as
$$\mathcal{P}=\frac{\rho}{2}\int_{t_1}^{t_2}\int_\Omega g[\eta^2-cD^2|\nabla \eta|^2]~d\bx~dt\ .$$
Then we define the action integral
$$\mathcal{I}=\mathcal{K}-\mathcal{P}+\rho\int_{t_1}^{t_2}\int_\Omega [\eta_t+\Div[(D+\eta)\bu]+\Div\left\{a\nabla(D^3\Div\bu)-bD^2\nabla\eta_t\right\}]\phi~d\bx~dt\ ,$$
where the mass conservation is imposed with the help of the Lagrange multiplier $\phi(\bx,t)$. Then the Euler-Lagrange equations for the action integral $\mathcal{I}$ are the following:
\begin{align}
\delta \phi~: & \quad \eta_t+\Div[(D+\eta)\bu]+\Div\left\{a\nabla(D^3\Div\bu)-bD^2\nabla\eta_t\right\}=0\ , \label{eq:var1}\\
\delta \bu~: & \quad \bu-\nabla\phi=0\ , \label{eq:var2}\\
\delta \eta~: & \quad \tfrac{1}{2}|\bu|^2-cg\Div(D^2\nabla\eta)-g\eta-\phi_t+b\nabla\cdot(D^2\nabla\phi_t)-\bu\cdot\nabla\phi=0\ . \label{eq:var3}
\end{align}
Eliminating $\phi$ in (\ref{eq:var3}) using (\ref{eq:var2}) we obtain the momentum conservation equation
$$\bu_t+g\nabla\eta+\tfrac{1}{2}\nabla|\bu|^2 +\nabla\left\{cg\Div(D^2\nabla\eta)-d\Div(D^2\bu_t)\right\}=0\ ,$$
with $b=d$. In this way we re-derived the subclass of systems (\ref{eq:Nwogunabcdd2}) that preserve a reasonable form of energy, and these must have necessarily $b=d$. From (\ref{eq:var2}) we observe that the new systems imply potential flow, and the conservation of the $\Curl\bu$ follows.

\section{Valid slip-wall boundary conditions}\label{sec:bcs}

Practical problems involving water waves are usually posed in bounded polygonal domains $\Omega$. For example, waves in a port or waves interacting with the sides of a basin. Wall boundary conditions must then be imposed alongside the model equations. When the waves can slip on the wall without any friction, a slip-wall boundary condition of the form $\bu\cdot\bn=0$ needs to be imposed. Here, $\bn$ is the outward unit normal vector on the boundary $\partial\Omega$. While no other boundary conditions are usually necessary, it is interesting to observe that the Euler equations implicitly obey additional boundary conditions. For example, the Neumann condition $\nabla\eta\cdot\bn=0$ holds true on $\partial\Omega$ when $\Omega$ is a polygonal domain, \cite{Khakimzyanov2018a}. To derive this particular boundary condition, first observe that $\nabla p\cdot\bn=0$ on $\partial\Omega$. This can be obtained by taking the inner product of equation (\ref{eq:momentv}) with $\bn$ on $\partial\Omega$. Writing the dynamic boundary condition (\ref{eq:pressbc}) as
$p(\bx,\eta(\bx,t),t)=p_{\text{atm}}$,
and taking the horizontal gradient operator, we obtain
$\nabla p+p_z\nabla \eta={\bf 0}$.
After multiplying by the normal vector $\bn$ and using the fact that $\nabla p\cdot\bn=0$ on $\partial\Omega$, we obtain the boundary condition $\nabla \eta\cdot \bn=0$ on $\partial\Omega$. If the domain is a general domain where well-posedness can be established, appropriate polygonal approximations $\Omega_n\subset\Omega$ can be constructed such that $\lim\Omega_n=\Omega$, and the particular Neumann condition can be used in more general situations. This particular Neumann condition is required for the well-posed Boussinesq systems of Bona-Smith and BBM-BBM type in bounded domains with slip-wall boundary conditions, and based on the previous analysis, it appears to be a physical condition \cite{KMS2020,IKKM2021}.  
Note that for the original Nwogu system (\ref{eq:Nwogu}), in addition to the slip-wall boundary conditions $\bu\cdot\bn=\nabla\eta\cdot\bn=0$ on $\partial\Omega$, it is required that $[\tilde{a}D^2\nabla(\Div (D\bu))+\tilde{b}D^3\nabla(\Div \bu)]\cdot\bn=0$ must be satisfied on the boundary $\partial\Omega$, \cite{WB2002}. This boundary condition is satisfied by the solutions of the Euler equations only when the bottom is flat, while for general bottoms, we obtain a simpler condition. 

From the irrotationality condition (\ref{eq:irrotational}), we have $\bu_z=\nabla w$, and thus $\nabla w\cdot \bn=0$ on $\partial \Omega$. Moreover, differentiating the mass equation (\ref{eq:mass}), we obtain $\nabla ( \Div \bu) =-\nabla w_z$, which implies the boundary condition $\nabla (\Div \bu )\cdot \bn =0$ on $\partial \Omega$. Similar conditions hold true for BBM-BBM-type systems. For example, the slip-wall conditions imply that $\nabla(\Div(D^2\bu_t))\cdot\bn=0$ on $\partial\Omega$. If the solution satisfies $\nabla(\Div D^2 \bu)\cdot\bn=0$ on $\partial \Omega$, then this is satisfied for $t\geq 0$. Due to the assumption of mild bottom variations in the derivation of the particular BBM-BBM system and also for the new Nwogu system, we have that $\nabla(\Div \bu)\cdot \bn\approx 0$, which means that even if the boundary condition is not satisfied exactly, it is a reasonable approximation of the exact boundary condition. The same will be true for Nwogu-type systems of the form (\ref{eq:Nwogunabcdd2}). This shows that some of the new $abcd$-systems (at least the BBM-BBM, Bona-Smith, Nwogu, and regularized Nwogu systems) can be used in practical situations in bounded domains with more accurate boundary conditions compared to the other Boussinesq systems. It is worth mentioning that of all the known Boussinesq systems, only the BBM-BBM system has been proved to be well-posed in bounded domains with slip-wall boundary conditions \cite{IKKM2021}, while the well-posedness in unbounded domains follows the work of \cite{DMS2007} in a similar manner. The Bona-Smith systems can also be proved to be well-posed in bounded domains, but this will be discussed in a separate work.

\section{Experimental validation}\label{sec:valid}

The mild slope assumption used in the derivation of the new $abcd$-Boussinesq systems with variable bottom topography raises a validity question: What are the minimum restrictions on the bottom variations to ensure the accuracy of the new systems? This question requires a more detailed study than the one presented here. However, in this section, we demonstrate that the particular assumption is not a significant barrier. On the contrary, the simplicity of the equations and their accuracy make them highly attractive for the mathematical modeling of long waves.

Among the systems we study in this section is the new Nwogu system. Although this system does not preserve any form of energy, it has optimal linear dispersion relation for $\theta^2=1/5$, $\mu=1$ and $\nu=0$ and outperforms other systems in experiments where the bottom variations as well as the steepness of the waves stress the validity of Boussinesq systems. 
For the sake of convenience we rewrite the new Nwogu system in the form
\begin{equation}\label{eq:Nwogunew}
\begin{aligned}
& \eta_t+\Div[(D+\eta)\bu]-\frac{1}{15}\Delta(D^3\Div\bu)=\frac{\sqrt{5}-2}{2\sqrt{5}}\Div(D^2\nabla\zeta_t)-\zeta_t\ ,\\
& \bu_t+g\nabla\eta+\tfrac{1}{2}\nabla|\bu|^2 -\frac{2}{5}\nabla\Div(D^2\bu_t)=\frac{\sqrt{5}-1}{\sqrt{5}}D\nabla\zeta_{tt}\ .
\end{aligned}
\end{equation}
 We also consider the new BBM-BBM system ($\theta^2=2/3$, $\nu=\mu=0$)
 \begin{equation}\label{eq:newbbmbbm}
\begin{aligned}
& \eta_t+\Div[(D+\eta)\bu]-\frac{1}{6}\Div (D^2\nabla\eta_t)=\frac{ 2-\sqrt{6}}{3}\Div(D^2\nabla\zeta_t)-\zeta_t\ ,\\
& \bu_t+g\nabla\eta+\tfrac{1}{2}\nabla|\bu|^2 -\frac{1}{6}\nabla\Div(D^2\bu_t)=\frac{\sqrt{3}-\sqrt{2}}{\sqrt{3}}D\nabla\zeta_{tt}\ ,
\end{aligned}
\end{equation}
which preserves the energy functional (\ref{eq:energyd}) with $a=c=0$. 
We also study the new Peregrine system ($\theta^2=1/3$, $\nu=\mu=0$)
\begin{equation}\label{eq:newperegrine}
\begin{aligned}
& \eta_t+\Div[(D+\eta)\bu]=\frac{\sqrt{3}-2}{2\sqrt{3}}\Div(D^2\nabla\zeta_t)-\zeta_t\ ,\\
& \bu_t+g\nabla\eta+\tfrac{1}{2}\nabla|\bu|^2 -\frac{1}{3}\nabla\Div(D^2\bu_t)=\frac{\sqrt{3}-1}{\sqrt{3}}D\nabla\zeta_{tt}\ .
\end{aligned}
\end{equation}
All the experiments primarily examine the impacts of bottom topography on long one-dimensional waves, which serve as standard benchmarks in the literature of Boussinesq systems. However, it should be noted that while the selected experiments somewhat deviate from the scope of Boussinesq systems, in the sense that the waves are of larger amplitude compared with what Boussinesq systems have been derived to describe, the results appear to be remarkably promising. Future work could focus on conducting more advanced experiments to further enhance our understanding.

For all the numerical simulations, we utilized previously developed Galerkin/Finite Element methods. For the Nwogu system, we employed the numerical method proposed in \cite{mm2023}. The method proposed by \cite{ADM2010} was utilized for the BBM-BBM system. Lastly, for the Peregrine system with variable bottom topography, we used the method introduced in \cite{AD2012}, which was extended accordingly. In all experiments, we employed cubic spline elements and the classical explicit fourth-order Runge-Kutta method with four stages.

\subsection{Shoaling and reflection of solitary waves}

In the first experiment, we investigate the capability of the new models to depict the shoaling behavior of waves on gentle slopes, as outlined in the study  \cite{Dodd1998}. Specifically, we focus on a one-dimensional computational domain $\Omega=[-100,20]$ with a bottom topography characterized by the following function: 
$$D(x)=\left\{ \begin{array}{ll}
0.7, & x\leq 0,\\
0.7-x/50, &x>0 .
\end{array}
\right.
$$
For the numerical simulation of this experiment, we employ wall boundary conditions at the endpoints $x=-100$ and $x=20$. All variables used in the simulation are expressed in SI units.

We begin by investigating the reflection of two right-traveling solitary waves with amplitudes $A_1=0.07~m$ and $A_2=0.12~m$. e two solitary waves have been horizontally shifted so that their maximum at $t=0$ is reached at $x_0=-30$. These waves initially propagate over the sloping bottom and then collide with the wall at $x=-20$, where they are reflected and propagate to the left. We record the free surface elevation at three wave gauges positioned at $x=0$, $x=16.25$, and $x=17.75$.

\begin{figure}[ht!]
  \centering
\includegraphics[width=0.75\columnwidth]{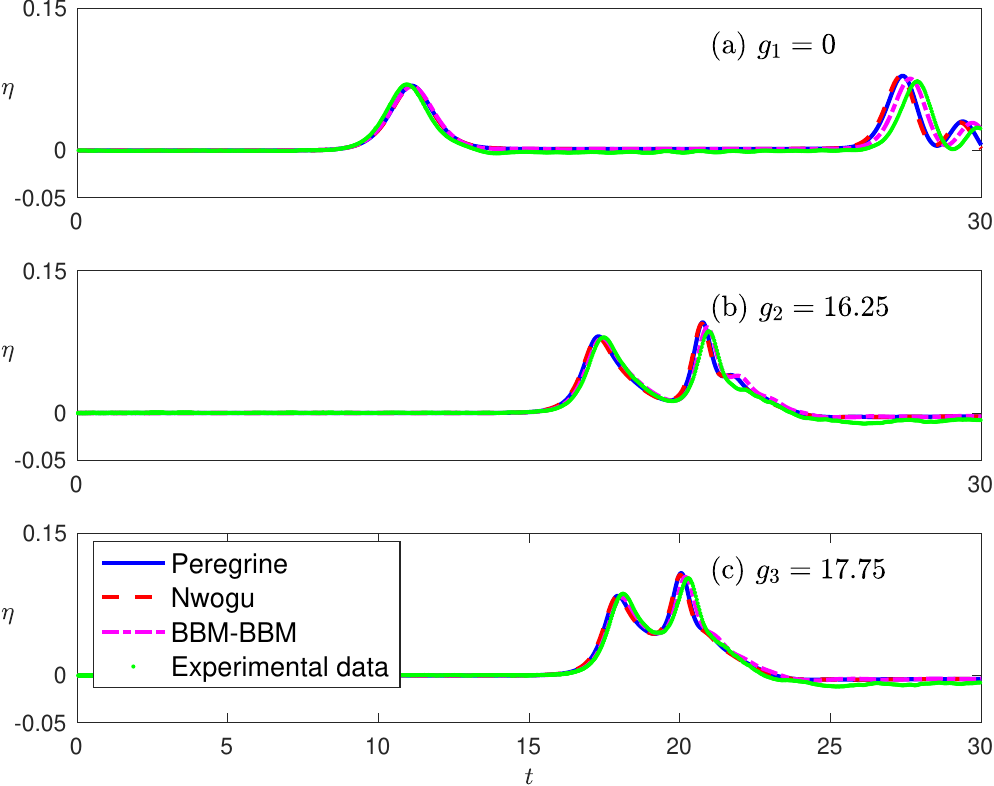}
  \caption{Reflection of shoaling solitary wave of amplitude $A=0.07~m$ by a vertical wall}
  \label{fig:shoal1}
\end{figure}
 \begin{figure}[ht!]
  \centering  \includegraphics[width=0.75\columnwidth]{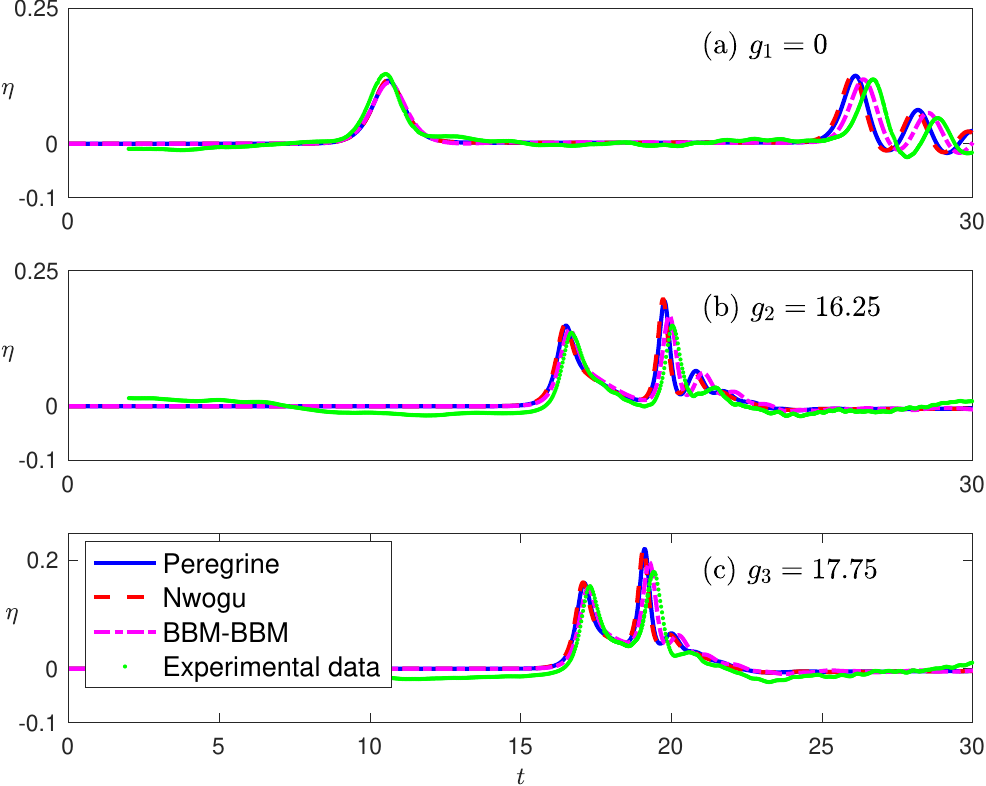}
  \caption{Reflection of shoaling solitary wave of amplitude $A=0.12~m$ by a vertical wall}
  \label{fig:shoal2}
\end{figure}

The results of these experiments for the solitary wave with $A_1=0.07~m$ are presented in Figure \ref{fig:shoal1}, while the results corresponding to the other solitary wave are shown in Figure \ref{fig:shoal2}. We observe that the solutions of the BBM-BBM system exhibit better agreement with the experimental data compared to the other systems. This behavior was also observed in \cite{KMS2020} with a BBM-BBM system that does not preserve any energy functional. In both cases, the values of $\varepsilon$ and $\sigma$ are small but not exactly within the range of validity of Boussinesq systems. For the first case, the characteristic values of the initial solitary wave are $\varepsilon\approx 0.23$ and $\sigma\approx 0.035$. For the second case, the parameters are $\varepsilon\approx 0.4$ and $\sigma\approx 0.05$. Moreover, the bottom variations parameter is $\beta=0.028$. It is important to note that these values change during the transformation of the shoaling solitary wave. Particularly, the value of $\varepsilon$ increases as the wave passes over the slope. This leads to a significant amount of error, which explains the deviations between the numerical and experimental results in Figures \ref{fig:shoal1} and \ref{fig:shoal2}.

We conclude the study of shoaling waves by examining an experiment conducted by \cite{GSSV1994}. In this particular experiment, our focus is on the shoaling behavior of a solitary wave with an amplitude of $A=0.2$ propagating over a bottom topography characterized by the following function:
$$D(x)=\left\{ \begin{array}{ll}
1, & x\leq 0,\\
1-x/35, &x>0 .
\end{array}
\right.
$$
We recorded the free surface elevation at five wave gauges that are located at $x=20.96$, $22.55$, $23.68$, $24.68$, $25.91$, respectively. It is important to note that in this experiment, all variables are scaled. This experiment proves to be more demanding compared to the previous one, requiring appropriate calibration of the location of the initial data in order to align the numerical solutions with the experimental data. This calibration is necessary because solitary waves of the same amplitude travel at different speeds for different Boussinesq systems \cite{DM2008}. As a result, solutions from different systems reach the wave gauges at different times.

To demonstrate this, let's examine the shoaling wave at the gauge located at $x=25.91$, which has an approximate amplitude of $a_0\approx 0.5$, while the depth is $D(x)\approx 0.26$, resulting in $\varepsilon\approx 1.92$. Although this value of $\varepsilon$ exceeds the range of validity for any Boussinesq system, the obtained results are still quite satisfactory after performing the appropriate calibration of the initial condition. The characteristic values for the initial solitary wave and the bottom in this case are $\varepsilon=0.2$, $\sigma\approx 0.05$, and $\beta\approx 0.028$.

Surprisingly, the BBM-BBM system required the least amount of phase calibration and exhibited the most accurate amplification of the solution among the tested systems. The results are illustrated in Figure \ref{fig:shoal3}. Once again, the BBM-BBM system demonstrates superior approximation of shoaling waves compared to the other systems. Similar observations were reported for the classical Nwogu system in \cite{FBCR2015}, emphasizing the need for initial condition calibration in all Boussinesq systems (except for the Serre-Green-Naghdi equations \cite{MSM2017}).

In the comparison between the classical Peregrine system (\ref{eq:Peregrin}) and the new Peregrine system, we found that we needed to apply the same calibration to the initial conditions for both systems in order to align the phase of the numerical solutions with the experimental data. However, when we utilized the Serre-Green-Naghdi system for the same experiment, there was no need for calibration of the initial data's location \cite{MSM2017}. This implies that even in extreme situations involving Boussinesq systems, the assumption of small bottom variations is not a disadvantage, as any Boussinesq system will inevitably lose accuracy in such extreme experiments. It is important to note that the phase error observed in the first experiment was negligible, and therefore, no calibration was necessary.
 \begin{figure}[ht!]
  \centering  \includegraphics[width=0.75\columnwidth]{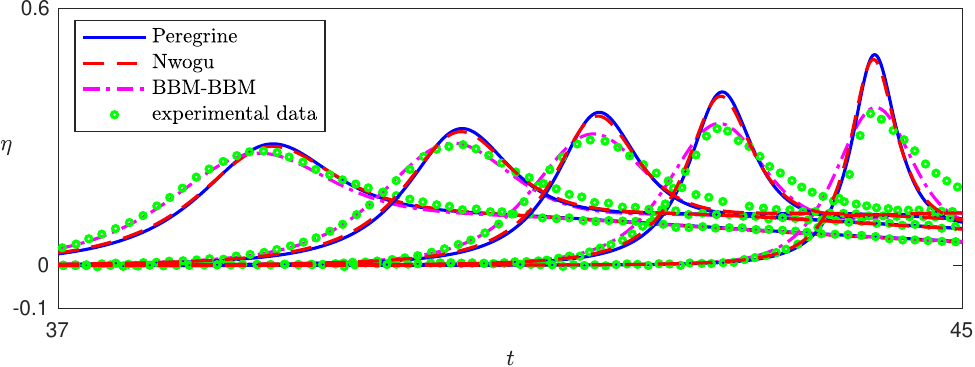}
  \caption{Recorded wave gauges of  a shoaling solitary wave of amplitude $A=0.2$ on a plain beach of slope $1~:~35$}
  \label{fig:shoal3}
\end{figure}

Please note that the numerical method used to solve the BBM-BBM system in our study is not conservative. However, we were able to measure the total energy of the solution (\ref{eq:energyd}), which was approximately $0.162$ for the first experiment, $0.362$ for the second experiment, and $0.135$ for the third experiment.

\subsection{Periodic waves over a submerged bar} 

In this scenario, we assess the performance of the aforementioned Boussinesq systems in a highly challenging laboratory experiment designed to investigate the nonlinear and dispersive properties of Boussinesq systems with variable bathymetry \cite{BB1994}. This specific experiment has been extensively used to validate different Boussinesq systems, including the original Nwogu system, and is regarded as a standard benchmark for numerical models \cite{Dingemans1994, KDNS12, WB1999}.

In this experiment, small-amplitude periodic waves are generated by a wavemaker and propagate over a bathymetry defined by the function:
$$
D(x,y)=\left\{
\begin{array}{cc}
-0.05x+0.7, & x\in[6,12) \\
0.1, & x\in[12,14) \\
0.1x-1.3, & x\in[14,17] \\
0.4, & \mbox{elsewhere}
\end{array}
 \right.
 $$
which is depicted in Figure \ref{fig:bottom}. 
 \begin{figure}[ht!]
  \centering
  \includegraphics[width=0.9\columnwidth]{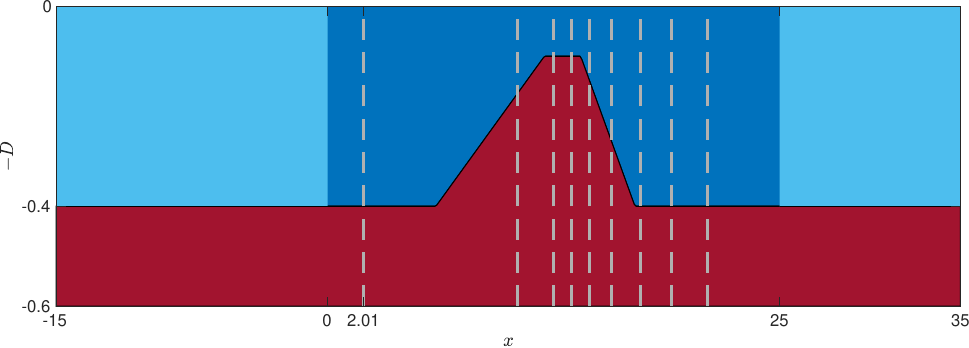}
  \caption{The bathymetry in relation to the sponge layers (light blue regions), wavemaker and wave gauges locations (distances in metres)}
  \label{fig:bottom}
\end{figure}

 \begin{figure}[ht!]
  \centering
  \includegraphics[width=\columnwidth]{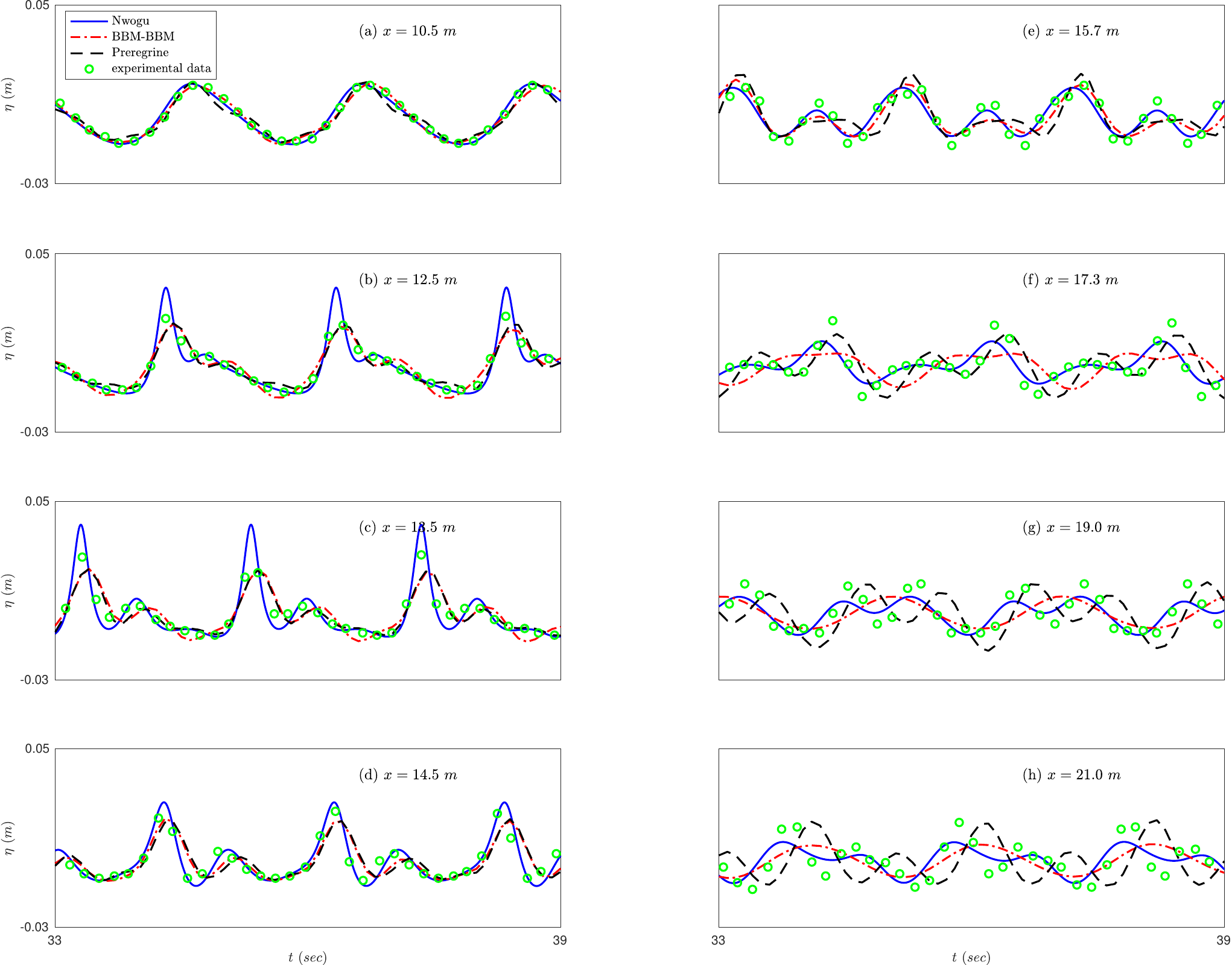}
  \caption{The recorded data at the various wave gauges and the numerical solutions for the new Nwogu, BBM-BBM and Peregrine systems}
  \label{fig:results}
\end{figure}

During the interaction of the generated waves with the submerged bar's uphill portion, they undergo shoaling and steepening. Subsequently, the waves propagate at an increased speed over the downhill part of the bar, encompassing a broad spectrum of long wave numbers. In this experiment, the input wave has a period of $\tau=2.02s$ and an amplitude of $A=0.01m$. The water's free surface is recorded at several wave gauges positioned at $x=10.5$, $12.5$, $13.5$, $14.5$, $15.7$, $17.3$, $19.0$, and $21.0$. For numerical simulations, the waves are generated using a wavemaker located at $x=2.01$, following the methodology described in \cite{KMS2020}.

The locations of the wavemaker and the wave gauges are depicted by vertical broken lines in Figure \ref{fig:bottom}. Experimental measurements of the free surface elevation at the different gauges are presented in Figure \ref{fig:results}. In the same figure, numerical results obtained using the standard Galerkin/Finite Element method with cubic splines, combined with the classical fourth-order Runge-Kutta method, are displayed for solving the one-dimensional initial-periodic boundary-value problem for the systems (\ref{eq:Nwogunew}), (\ref{eq:Nwogunabcdd2}) corresponding to the new Nwogu system, BBM-BBM system, and Peregrine system. Further details regarding the specific numerical method can be found in \cite{ADM2010ii} and \cite{mm2023}. It is noteworthy that the numerical solution exhibits an impressive level of agreement with the experimental data.

When comparing the results obtained using the new BBM-BBM, Peregrine and Nwogu systems, employing similar numerical methods and wave generation procedures as described in \cite{KMS2020}, we observe that the performance of the new Nwogu system is superior due to its optimal dispersion properties. Typically, numerical solutions of other Boussinesq systems, including Peregrine's system, start deviating from the experimental data after wave gauge (f). In contrast, the phase of the numerical solution for the new Nwogu system agrees with the experimental data across all wave gauges. Other properties, such as the shoaling of solitary waves, have been tested for the new BBM-BBM and Peregrine systems in \cite{IKKM2021}. For this specific experiment, we have estimated that $\varepsilon\approx 0.5$ and $\sigma\approx 0.1$. These values indicate that the experiment is situated outside the theoretical range of validity for Boussinesq systems.

In conclusion, the simplifications made by assuming mild bottom variations do not compromise the validity of Boussinesq systems, even in highly demanding experiments involving nearly breaking waves, as demonstrated in this section.

\section{Conclusions}
The class of $abcd$-Boussinesq systems for water wave theory, introduced in \cite{BCS2002} and \cite{BCL2005}, is generalized to a new set of $abcd$-Boussinesq systems (\ref{eq:Nwogunabcdd2})--(\ref{eq:abcdcoef2}) in the case of variable bottom topography and in two-dimensional domains. These systems have been formulated to be applicable in simulations with slip-wall boundary conditions, considering practical applications. We also demonstrate that the new systems can accurately accommodate slip-wall boundary conditions, which was not straightforward with their classical counterparts. A specific subclass of the new systems preserves a meaningful energy functional, making them suitable for conservative numerical simulations. The consistency of the new systems with the Euler equations has also been studied, and the errors between the Euler equations and the new Boussinesq systems are estimated, extending the results of \cite{Lannes13,BCL2005}. Lastly, we show that the assumption of smooth bathymetric variations is not restrictive when it comes to classical benchmarks. These findings align with the conclusions of the work in \cite{MS1991}.


\end{document}